\documentclass[a4paper,twocolumn,pre,showpacs,showkeys,aps]{revtex4-1}
\usepackage{epsfig,amsmath,amsfonts,amssymb}
\usepackage{amsmath,amsfonts,amssymb}
\usepackage{graphicx}
\usepackage{subfigure}
\usepackage{mathbbol}
\usepackage{amsthm}

\newtheorem{theorem}{Theorem}
\newtheorem{prob}{Problem}
\newtheorem{proposition}{Proposition}

\newtheorem{remark}{Remark}
\newtheorem{lemma}{Lemma}

\begin{document}

\title{An Exponential Bound in the Quest for Absolute Zero}

\author{Dionisis Stefanatos}
\email{dionisis@post.harvard.edu}
%\altaffiliation[Current address:]{ 3 Omirou St., Sami, Kefalonia 28080, Greece.}
%\email{dionisis@post.harvard.edu}
\affiliation{Division of Physical Sciences and Applications,\\
  Hellenic Army Academy, Vari, Athens 16673, Greece}

\date{\today}% It is always \today, today,

\begin{abstract}
In most studies for the quantification of the third thermodynamic law, the minimum temperature which can be achieved with a long but finite-time process
scales as a negative power of the process duration. In this article, we use our recent complete solution for the optimal control problem of the quantum parametric oscillator to show that the minimum temperature which can be obtained in this system scales exponentially with the available time. The present work is expected to motivate further research in the active quest for absolute zero.
\end{abstract}

%\pacs{05.45.Xt}% PACS, the Physics and Astronomy
%Use showkeys class option if keyword
\maketitle

\section{Introduction}

The third law of Thermodynamics, as formulated by W. Nernst in the form of the unattainability principle, states that no process can reduce the temperature of any system to absolute zero in a finite number of steps and within finite time \cite{Nernst12}. Several studies have been devoted to the quantification of this principle, i.e. finding the minimum temperature which can be obtained as a function of the duration of the cooling process \cite{Rezek09,Hoffmann_EPL11,Levy12,TorronKos13,Masanes17}. The minimum temperature found in most of these works scales as an inverse power of the available time.

In the present article, we use our recent complete solution for the optimal control problem of the quantum parametric oscillator \cite{Stefanatos2017a} and derive an exponential bound on the minimum achievable temperature with this system. Note that such a bound has been previously obtained for this system, but in the more general case where the harmonic potential can become repulsive \cite{Hoffmann_EPL11}. Here we prove the exponential bound for the more restrictive and practically interesting case where the stiffness of the parametric oscillator remains positive throughout.

%%%%%%%%%%%%%%%%%%%%%%%%%%%%%%%%%%%%%%%%%%
\section{The quest for absolute zero using the quantum parametric oscillator}

The system that we consider in this article is an ensemble of noninteracting particles of mass $m$ trapped in a parametric harmonic potential \cite{Rezek09,Hoffmann_EPL11,Rezek06,Stefanatos2017a,Salamon09,Tsirlin11,Salamon12,Hoffmann13,Hoffmann15,Boldt16,Stefanatos2017b}. The corresponding Hamiltonian is
\begin{equation}
\label{Hamiltonian}
\hat{H}=\frac{\hat{p}^2}{2m}+\frac{m\omega^2(t)\hat{q}^2}{2},
\end{equation}
where $\hat{q}, \hat{p}$ are the position and momentum operators, respectively, and $\omega(t)$ is the time-varying frequency of the oscillator which serves as the available control and is restricted between a maximum and a minimum value
\begin{equation}
\label{frequency}
\omega_c\leq\omega(t)\leq\omega_h.
\end{equation}
This system can be considered as the quantum analog of a classical piston, with the frequency corresponding to the inverse volume of the classical case, since a larger frequency results to a tighter confinement of the particle wavefunction in space. This \emph{quantum piston} has been extensively used in the quest for absolute zero and other thermodynamic applications \cite{Chen10,Zulkowski12,shortcuts13,Deng13,Deffner13,Campo14,Abah16}; in this section we summarize some well known facts from the related literature, which can also be found in the recent review \cite{Kosloff17}.

We first describe the dynamics under Hamiltonian (\ref{Hamiltonian}). Recall from Quantum Mechanics that the time evolution of a quantum observable (hermitian operator) $\hat{O}$ in the Heisenberg picture is given by \cite{Merzbacher98}
\begin{equation}
\label{Observable}
\frac{d\hat{O}}{dt}=\frac{i}{\hbar}[\hat{H},\hat{O}]+\frac{\partial\hat{O}}{\partial t},
\end{equation}
where $\imath=\sqrt{-1}$ and $\hbar$ is Planck's constant.
The following operators form a closed Lie algebra \cite{Boldt16}
\begin{equation}
\label{z_operators}
\hat{z}_1=m\hat{q}^2,\quad \hat{z}_2=\frac{\hat{p}^2}{m},\quad \hat{z}_3=-\frac{\imath}{2\hbar}[\hat{z}_1,\hat{z}_2]=\hat{q}\hat{p}+\hat{p}\hat{q},
\end{equation}
while Hamiltonian (\ref{Hamiltonian}) is a linear combination of $\hat{z}_1$ and $\hat{z}_2$. As explained in \cite{Rezek06}, the state of the system under
evolution (\ref{Hamiltonian}) can be described by the expectation values
\begin{equation}
\label{expectations}
z_i=\langle\hat{z}_i\rangle=\mbox{Tr}(\rho_0\hat{z}_i),\quad i=1, 2, 3
\end{equation}
of these operators, where $\rho_0$ is the density matrix corresponding to the initial state of the system at $t=0$ (recall that we use the Heisenberg picture). The explicit relation between the state of the system (density matrix) and these three observables can be found in \cite{Rezek06}. From (\ref{Observable}) and (\ref{expectations}) we easily find the equations
%\begin{align}
%\label{z1}\dot{z}_1  &  = z_3,\\
%\label{z2}\dot{z}_2  &  = -\omega^2z_3,\\
%\label{z3}\dot{z}_3  &  = -2\omega^2z_1+2z_2.
%\end{align}
\begin{equation}
\label{z}
\dot{z}_1 = z_3,\quad\dot{z}_2 = -\omega^2z_3,\quad\dot{z}_3 = -2\omega^2z_1+2z_2.
\end{equation}

We consider that initially (for time $t\leq 0$) the system is in thermal equilibrium with a hot bath at temperature $T_h$, while the frequency is fixed to the maximum allowed value $\omega_h$. In order to find the initial values of $z_i$ note that states of thermodynamic equilibrium, with $\omega(t)=\omega$ constant, are characterized by the equipartition of average energy $E=\langle\hat{H}\rangle$
\begin{equation}
\label{equipartition}
\left\langle\frac{\hat{p}^2}{2m}\right\rangle=\left\langle\frac{m\omega^2\hat{q}^2}{2}\right\rangle=\frac{E}{2}
\end{equation}
and the absence of correlations
\begin{equation}
\label{no_correlation}
\langle\hat{q}\hat{p}+\hat{p}\hat{q}\rangle=0.
\end{equation}
Starting at $t=0$ from the equilibrium state with temperature $T_h$ and frequency $\omega_h$, the corresponding average energy $E_0$ is
\[
E_0=\frac{\hbar\omega_h}{2}\coth\left(\frac{\hbar\omega_h}{2k_bT_h}\right),
\]
where $k_b$ is Boltzmann's constant. Using (\ref{equipartition}) and (\ref{no_correlation}) in (\ref{z_operators}) we find the initial conditions
%\begin{equation}
%\label{z_initial}
%z_1(0)=\frac{E_0}{\omega_0^2},\quad z_2(0)=E_0,\quad z_3(0)=0.
%\end{equation}
\begin{equation}
\label{z_initial}
z_1(0)=E_0/\omega_h^2,\quad z_2(0)=E_0,\quad z_3(0)=0.
\end{equation}

At $t=0$ the system is isolated from the hot reservoir and $\omega(t)$ is varied within the interval (\ref{frequency}) until it reaches the minimum allowed value $\omega(\tau)=\omega_c$ for some final time $t=\tau$ yet unspecified, thus the frequency boundary conditions are
\begin{equation}
\label{frequency_boundary}
\omega(t)=\left\{\begin{array}{cl} \omega_h, & t\leq 0 \\\omega_c, & t\geq \tau\end{array}\right.
\end{equation}
As we explained above, the frequency of the potential corresponds to the inverse volume of a classical piston. If the frequency is appropriately reduced to its final value (corresponding to a classical increase in volume) while the system is isolated and does not exchange heat with its environment, it is generally expected the cooling of the trapped particles. In order to find the minimum achievable temperature with this process, we find the minimum possible energy at the final time. During the expansion of the quantum piston, the evolution is governed by the system (\ref{z}). As it can be easily verified, the following quantity, called the Casimir companion, is a constant of the motion \cite{Boldt13}
%\begin{equation}
%\label{constant}
%z_1z_2-\frac{z_3^2}{4}=\frac{E_0^2}{\omega_0^2}.
%\end{equation}
\begin{equation}
\label{constant}
z_1z_2-z_3^2/4=E_0^2/\omega_h^2.
\end{equation}
The instantaneous average energy $E=\langle\hat{H}\rangle$ during the expansion can be expressed as
\begin{equation}
\label{average_energy}
E=\frac{1}{2}(\omega^2z_1+z_2)
\end{equation}
and, if we solve with respect to $z_2$ and replace it in (\ref{constant}) we obtain
\[
\omega^2z_1^2-2Ez_1+\frac{z_3^2}{4}+\frac{E_0^2}{\omega_h^2}=0.
\]
This equation has real solutions with respect to $z_1$ when
\[
E\geq\sqrt{\frac{1}{4}\omega^2z_3^2+\frac{\omega^2}{\omega_h^2}E_0^2},
\]
thus
\[
E\geq\frac{\omega}{\omega_h}E_0\geq\frac{\omega_c}{\omega_h}E_0,
\]
where the equality holds for $z_3=0$ and $\omega=\omega_c$. The energy $E_f$ at the final time $t=\tau$, where $\omega=\omega_c$, achieves the lower bound for $z_3(\tau)=0$. From this last relation and (\ref{constant}), (\ref{average_energy}) we obtain the terminal conditions at the final time $t=\tau$
%\begin{equation}
%\label{z_final}
%z_1(T)=\frac{E_f}{\omega_f^2},\quad z_2(T)=E_f,\quad z_3(T)=0.
%\end{equation}
\begin{equation}
\label{z_final}
z_1(\tau)=E_f/\omega_c^2,\quad z_2(\tau)=E_f,\quad z_3(\tau)=0,
\end{equation}
corresponding to the minimum value for the final energy $E_f$
\begin{equation}
\label{final_energy}
E_f=\frac{\omega_c}{\omega_h}E_0=\frac{\hbar\omega_c}{2}\coth\left(\frac{\hbar\omega_h}{2k_bT_h}\right).
\end{equation}
Observe from (\ref{z_final}) that the final state is also in thermal equilibrium, and from (\ref{final_energy}) the corresponding internal or effective temperature $T_c$ can be identified as
\begin{equation}
\label{T_c}
T_c=\frac{\omega_c}{\omega_h}T_h.
\end{equation}
Since $\omega_c<\omega_h$, this obviously corresponds to the cooling of the trapped particles.

Having determined the minimum achievable temperature under condition (\ref{frequency}), it is then natural to ask how should we choose $\omega(t)$ in order to reach it. The trivial answer is to decrease the frequency from the initial value $\omega_h$ to the final $\omega_c$ following a slow (adiabatic) process. In this case $z_3=0$ throughout the process and the system moves along thermal equilibrium states which, according to (\ref{constant}), lie on the hyperbola
\[
z_1z_2=E_0^2/\omega_h^2
\]
in the $z_1z_2$-plane. From this relation we find that the instantaneous average energy $E$ and frequency $\omega$ satisfy
\[
\label{energy_ratio}
E=\frac{\omega}{\omega_h}E_0,
\]
which is nothing more than the well known adiabatic invariant $J=E/\omega$ of the harmonic oscillator. At the final time $t=\tau$, where $\omega=\omega_c$, the desired minimum energy (\ref{final_energy}) and temperature (\ref{T_c}) are obtained.

The problem with the slow adiabatic process is that it requires long times, an undesirable characteristic which makes it impractical, thus alternative approaches are needed. In a highly influential paper \cite{Salamon09}, the authors suggested to reach the minimum temperature with the following finite-time process
\begin{equation}
\label{frequency_profile}
\omega(t)=\left\{\begin{array}{cl} \omega_c, & 0<t\leq \tau_c \\\omega_h, & \tau_c<t<\tau=\tau_c+\tau_h\end{array}\right.,
\end{equation}
with duration
\begin{eqnarray*}
\label{one_intermediate}
\tau &=& \tau_c+\tau_h\nonumber\\
&=&\frac{1}{2\omega_c}\cos^{-1}\left[\frac{\omega_c^2+\omega_h^2}{(\omega_c+\omega_h)^2}\right]+\frac{1}{2\omega_h}\cos^{-1}\left[\frac{\omega_c^2+\omega_h^2}{(\omega_c+\omega_h)^2}\right].
\end{eqnarray*}
Note that there is an intermediate switching at $t=\tau_c$ from $\omega_c$ to $\omega_h$, while there are also instantaneous jumps at the initial and final times, in order to satisfy the boundary conditions (\ref{frequency_boundary}). In the limit $\omega_c\rightarrow 0$ and for finite $\omega_h$ the duration of the process approaches
\begin{equation}
\label{tau_limit}
\tau\rightarrow\frac{1}{\sqrt{\omega_h\omega_c}}.
\end{equation}
From (\ref{T_c}) observe that $T_c\sim\omega_c$, thus in the limit $\omega_c\rightarrow 0$ and for finite $\omega_h,T_h$ it is also $T_c\rightarrow 0$, while
\begin{equation}
\label{min_time_1}
\tau\rightarrow \frac{1}{\omega_h}\sqrt{\frac{T_h}{T_c}}
\end{equation}
and
\begin{equation}
\label{min_energy_1}
T_c\rightarrow \frac{T_h}{\omega_h^2\tau^2}.
\end{equation}
Eq. (\ref{min_time_1}) expresses how the process duration diverges when the temperature $T_c$ approaches the absolute zero, while Eq. (\ref{min_energy_1}) indicates the minimum temperature that can be achieved by a process of the form (\ref{frequency_profile}) with duration $\tau$.

The pulse sequence (\ref{frequency_profile}) provides the minimum-time cooling solution for small values of the frequency ratio $\omega_h/\omega_c>1$. For larger values of this ratio, we recently showed in \cite{Stefanatos2017a} that there might be pulse sequences with more intermediate switchings which achieve faster cooling, and provided a specific such example where the optimal solution contains three switchings. In the next section we use our complete solution for the optimal control problem of the quantum parametric oscillator \cite{Stefanatos2017a}, and obtain a stricter logarithmic bound for the cooling time in the limit $\omega_c\rightarrow 0$, where the ratio $\omega_h/\omega_c$ becomes large for finite $\omega_h$.

%%%%%%%%%%%%%%%%%%%%%%%%%%%%%%%%%%%%%%%%%%
\section{Implications of the minimum-time solution}

We start by presenting the minimum-time solution given in \cite{Stefanatos2017a}.
If we define the dimensionless variable $b$ through the relations
\[
\label{b}
b=\frac{\sqrt{\langle\hat{q}^2\rangle}}{q_0},\quad q_0=\sqrt{\frac{E_0}{m\omega_h^2}},
\]
where $q_0$ has length dimensions, then the expectation values $z_i$ can be expressed in terms of $b$ and its derivatives as follows
\begin{equation}
\label{zeta}z_1=\frac{E_0}{\omega_h^2}b^2,\quad z_2=\frac{E_0}{\omega_h^2}(b\ddot{b}+\dot{b}^2+\omega^2b^2),\quad z_3=\frac{2E_0}{\omega_h^2}b\dot{b}.
\end{equation}
If we plug (\ref{zeta}) in (\ref{constant}), we obtain the following Ermakov equation for $b(t)$ \cite{Chen10}
\begin{equation}
\label{Ermakov}\ddot{b}(t)+\omega^{2}(t)b(t)=\frac{\omega_{h}^{2}}{b^{3}(t)}%
\end{equation}
The boundary conditions for $b$ can be found by using (\ref{zeta}) in (\ref{z_initial}) and (\ref{z_final}). They are
\begin{equation}
\label{b_boundary}b(0)=1,\quad \dot{b}(0)=0,\quad b(\tau)=\sqrt{\omega_h/\omega_c},\quad \dot{b}(\tau)=0,
\end{equation}
where we have additionally used (\ref{final_energy}) in the derivation of $b(\tau)$.

If we set
\begin{equation}
\label{x}
x_{1}=b,\quad x_{2}=\frac{\dot{b}}{\omega_{h}},\quad u(t)=\frac{\omega^{2}(t)}%
{\omega_{h}^{2}},
\end{equation}
and rescale time according to $t_{\mbox{new}}=\omega_{h} t_{\mbox{old}}$ while keeping the same notation for the normalized time, we
obtain the following system of first order differential equations, equivalent
to the Ermakov equation %(\ref{Ermakov})
\begin{eqnarray}
\label{system1}\dot{x}_{1}  &  =  & x_{2},\\
\label{system2}\dot{x}_{2}  &  =  & -ux_{1}+\frac{1}{x_{1}^{3}}.
\end{eqnarray}
The control bounds are $u_1=\omega_c^2/\omega_h^2, u_2=\omega_h^2/\omega_h^2=1$,
and if we set
\begin{equation}
\label{gamma}
\gamma=\sqrt{\frac{\omega_h}{\omega_c}}>1
\end{equation}
they become
\begin{equation}
\label{u_order}
u_1=\frac{1}{\gamma^4}<1,\quad u_2=1.
\end{equation}
Using (\ref{x}) to translate the boundary conditions (\ref{b_boundary}) for $b$ into corresponding conditions for $x_1, x_2$, we obtain the following time-optimal problem for system (\ref{system1}), (\ref{system2}):
\begin{prob}\label{problem1}
Find $u_1\leq u(t)\leq u_2$, with $u_1=1/\gamma^4$ and $u_2=1$, such that starting from $(x_1(0),x_2(0))=(1,0)$, the system above reaches the final point $(x_1(\tau),x_2(\tau))=(\gamma,0), \gamma>1$, in minimum time $\tau$.
\end{prob}

In our recent work \cite{Stefanatos2017a}, we solved this problem for the more general case where $u_1\leq 1/\gamma^4$ and $u_2\geq 1$. Here we present the solution when the control bounds are fixed as in (\ref{u_order}), corresponding to the bounds in (\ref{frequency}).
\begin{theorem}
\label{solution}
The optimal control $u(t)$ has the bang-bang form, i.e. alternates between the boundary values $u_1$ and $u_2$, with an odd number of switchings, starting with $u=u_1$ and ending with $u=u_2$. The ratio of the coordinates $(x_2/x_1)$ of consecutive switching points has constant magnitude but alternating sign, while these points are not symmetric with respect to $x_1$-axis. The square of this ratio, $s=x_2^2/x_1^2$, is obviously constant at the switching points. The necessary time to reach the
target point $(\gamma,0), \gamma>1$, with a candidate optimal trajectory (extremal) with $2n+1$ switchings, $n=0,1,2,\ldots$, is
\begin{equation}
\label{time_odd}
\tau^{\pm}_{2n+1}=\tau^{\pm}_{i}+n(\tau_{u_1}+\tau_{u_2})+\tau_{f},
\end{equation}
where
\begin{eqnarray}
\label{time_in1}\tau^{\pm}_{i}  &  = & \frac{1}{2\sqrt{u_{1}}}\cos^{-1}\left[
\frac{sc_{1}\mp u_{1}\sqrt{c_{1}^{2}-4(s+u_{1})}}{(s+u_{1})\sqrt{c_{1}%
^{2}-4u_{1}}}\right],\\
\label{time_fi}
\tau_{f}  &  = & \frac{1}{2\sqrt{u_{2}}}\cos^{-1}\left[  \frac{-sc+u_{2}%
\sqrt{c^{2}-4(s+u_{2})}}{(s+u_{2})\sqrt{c^{2}-4u_{2}}}\right]  ,
\end{eqnarray}
\begin{eqnarray}
\label{switch_X}
\tau_{u_1}  &  = & \frac{1}{2\sqrt{u_{1}}}\cos^{-1}\left(
\frac{s-u_{1}}{s+u_{1}}\right),\\
\label{switch_Y}
\tau_{u_2}  &  = & \frac{1}{2\sqrt{u_{2}}}\left[  2\pi-\cos^{-1}\left(  \frac
{s-u_{2}}{s+u_{2}}\right)  \right]  ,
\end{eqnarray}
\begin{eqnarray}
\label{c_1}
c_{1}  &  = & u_{1}+1,\\
\label{c}
c  &  = & u_{2}\gamma^{2}+\frac{1}{\gamma^{2}},
\end{eqnarray}
and  the ratio $s$ satisfies the transcendental equation
\begin{equation}
\label{transcendentalX}\frac{c+\sqrt{c^{2}-4(s+u_{2})}}{c_{1}%
\pm\sqrt{c_{1}^{2}-4(s+u_{1})}}=\left(\frac{s+u_{2}}{s+u_{1}}\right)^{n+1}%
\end{equation}
in the interval $0<s\leq s_m$, where
\begin{equation}
s_m=\frac{1}{4}(1-u_{1})^{2}.
\end{equation}
Note that the $\pm$ sign in (\ref{transcendentalX}) corresponds to the $\pm$ sign in (\ref{time_odd}). The constants $c_{1}$ and
$c$ characterize the first and the last segments, respectively, of
the trajectory.
\end{theorem}

\begin{proof}
Theorem \ref{solution} is actually Theorem 2 from \cite{Stefanatos2017a}, with the appropriate modifications to account for the specific values of the control bounds given in (\ref{u_order}). Here we highlight the most important points. For $u_2=1$, the starting point $(1,0)$ is an equilibrium point for system (\ref{system1}), (\ref{system2}), thus the candidate optimal trajectories (extremals) should start with $u=u_1$. The solution of the transcendental equation in Theorem 2 from \cite{Stefanatos2017a} is restricted in the interval $0<s\leq\mbox{Min}\{(1-u_{1})^{2}/4, (u_{2}\gamma^2-1/\gamma^2)^{2}/4\}$, which is simplified to $0<s\leq (1-u_{1})^{2}/4$ for the control bounds given in (\ref{u_order}).
\end{proof}

\begin{remark}
We emphasize that the above described candidate optimal control $u(t)$ does not include the jumps at the initial and final times, which should be included in order to satisfy the boundary conditions (\ref{frequency_boundary}) for the frequency.
\end{remark}

In Fig. \ref{extremal} we depict a typical extremal trajectory, where the blue solid line corresponds to segments with $u(t)=u_1$ and the red dashed line to segments with $u(t)=u_2$. Times $\tau_i$ and $\tau_f$ are the times spent on the initial blue and final red segments, respectively, while $\tau_{u_1}$ $(\tau_{u_2}$) is the time spent on each intermediate blue (red) segment, where $u=u_1$ ($u=u_2$). The physical intuition behind these ``spiral" optimal solutions can be understood if Eqs. (\ref{system1}), (\ref{system2}) are interpreted as describing the one-dimensional Newtonian motion of a unit mass particle, with $x_1$ being the displacement and $x_2=\dot{x}_1$ the velocity. Under this point of view, observe from (\ref{system2}) that there is a strong repulsive force as $x_1\rightarrow0$ ($1/x_1^3$), which can act as a slingshot and push the particle faster to the target point.
\begin{figure}[t]
\centering
\includegraphics[width=8 cm]{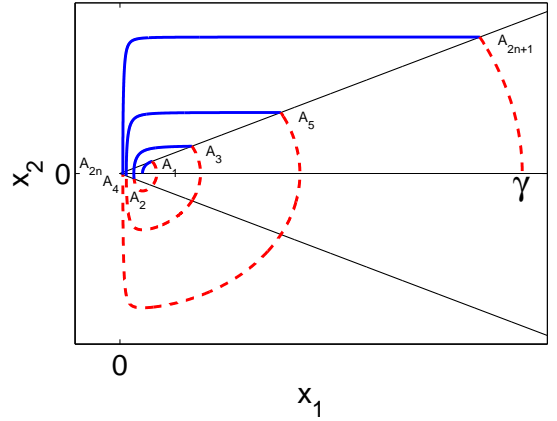}
\caption{A typical extremal trajectory. Blue solid line corresponds to segments with constant control $u=u_1$ and red dashed line to segments with $u=u_2$. The switching points lie on two straight lines passing through the origin, which are symmetric with respect to $x_1$-axis.}
\label{extremal}
\end{figure}

In this section we will use Theorem \ref{solution} to study the behavior of the extremal times in the limit $\omega_c\rightarrow 0$, corresponding to $\gamma\rightarrow\infty$. We will need the following lemma, regarding the monotonicity of the left and right hand sides of the transcendental equation (\ref{transcendentalX}).

\begin{lemma}
\label{monotonicity}
Let $l_{\pm}(s), r_n(s)$ denote the left and right hand sides of the transcendental equation (\ref{transcendentalX}). Then, $l_{-}(s)$ and $r_n(s)$ are decreasing functions of $s$, while $l_{+}(s)$ is increasing.
\end{lemma}

\begin{proof}
We find the derivatives of these functions with respect to $s$. It is
\begin{equation}
r_n'(s)=(n+1)\frac{u_1-u_2}{(s+u_1)^2}\left(\frac{s+u_2}{s+u_1}\right)^n<0,\nonumber
\end{equation}
since $u_1<u_2$, and
\begin{eqnarray*}
l'_{\pm}(s) & = & \frac{2}{\left[c_1\pm\sqrt{c_1^2-4(s+u_1)}\right]^2}\\
 & &\times\left[\pm\frac{c+\sqrt{c^2-4(s+u_2)}}{\sqrt{c_1^2-4(s+u_1)}}-\frac{c_1\pm\sqrt{c_1^2-4(s+u_1)}}{\sqrt{c^2-4(s+u_2)}}\right].
\end{eqnarray*}
Obviously it is $l'_{-}(s)<0$, while we will show that $l'_{+}(s)>0$. Observe that
\begin{equation}
\sqrt{c^2-4(s+u_2)}>\sqrt{c_1^2-4(s+u_1)}\Leftrightarrow (\gamma^2-1)(\gamma^4-1)>0,\nonumber
\end{equation}
where the last inequality is true since $\gamma>1$, thus the first inequality also holds. From this inequality we obtain
\begin{equation}
\frac{1}{\sqrt{c_1^2-4(s+u_1)}}>\frac{1}{\sqrt{c^2-4(s+u_2)}}\nonumber
\end{equation}
and, since $c>c_1$
\begin{equation}
c+\sqrt{c^2-4(s+u_2)}>c_1+\sqrt{c_1^2-4(s+u_1)}.\nonumber
\end{equation}
By combining the last two inequalities we find
\begin{equation}
\frac{c+\sqrt{c^2-4(s+u_2)}}{\sqrt{c_1^2-4(s+u_1)}}>\frac{c_1+\sqrt{c_1^2-4(s+u_1)}}{\sqrt{c^2-4(s+u_2)}},\nonumber
\end{equation}
thus $l'_{+}(s)>0$.
\end{proof}

We begin our study from the one-switching solution and prove the following result.

\begin{proposition}
\label{one_switching}
There is only one extremal with only one switching ($n=0$). For large values of $\gamma$, the ratio $s$ corresponding to this extremal can be found by solving the transcendental equation (\ref{transcendentalX}) with the $+$ sign. In this limit, the corresponding time to reach the final point behaves as
\begin{equation}
\label{previous_bound}
\tau^+_1\rightarrow\gamma.
\end{equation}
\end{proposition}

\begin{proof}
The one-switching extremal is composed by two segments, one with $u=u_1$ and one with $u=u_2$. From system equations (\ref{system1}), (\ref{system2}) and the fact that the initial and final points, $(1,0)$ and $(\gamma,0)$, belong to the first and the second segments, respectively, we can easily obtain the equations of these segments
\begin{eqnarray*}
x_2^2+u_1x_1^2+\frac{1}{x_1^2}&=&c_1,\\
x_2^2+u_2x_1^2+\frac{1}{x_1^2}&=&c,
\end{eqnarray*}
where the constants $c,c_1$ are given in Theorem \ref{solution}. The above system of equations has only one solution
\begin{equation}
\label{switching_point}
x_1=\sqrt{\frac{c-c_1}{u_2-u_1}},\quad x_2=\sqrt{c-u_2x_1^2-\frac{1}{x_1^2}},
\end{equation}
since it is always $x_1>0$ when starting from $x_1(0)=1>0$, while only the solution $x_2>0$ is acceptable for motion from $(1,0)$ to $(\gamma,0)$. Depending on the value of $\gamma$, this unique solution can be the solution of the transcendental equation (\ref{transcendentalX}) with $n=0$ and only the $+$ or only the $-$ sign. We will prove that, for large values of $\gamma$, it is the solution of the transcendental equation with the $+$ sign. We evaluate the left and right hand sides of (\ref{transcendentalX}), with $n=0$ and the $+$ sign, at the bounds of the interval $0<s\leq s_m=(1-u_{1})^{2}/4$, and find
\[
r_0(0)=\frac{u_2}{u_1}=\gamma^4,\quad r_0(s_m)\rightarrow 5,
\]
and
\[
l_+(0)=\gamma^2,\quad l_+(s_m)\rightarrow 2\gamma^2,
\]
where the limits correspond to large values of $\gamma$. Observe that $r_0(0)>l_+(0)$, while for large $\gamma$ it is also $r_0(s_m)<l_+(s_m)$. Since $r_0(s), l_+(s)$ are decreasing and increasing, respectively, continuous functions of $s$, the corresponding transcendental equation has a unique solution in the above interval, corresponding to the unique solution with one switching. By either manipulating the transcendental equation (\ref{transcendentalX}) with $n=0$ or using directly the coordinates of the switching point given in (\ref{switching_point}), we find the solution
\[
s=\frac{x_2^2}{x_1^2}=\frac{c_1u_2-cu_1}{c-c_1}-\left(\frac{u_2-u_1}{c-c_1}\right)^2,
\]
from which we obtain
\begin{equation}
\label{s_limit_0}
s\rightarrow\frac{1}{\gamma^2}-\frac{1}{\gamma^4}
\end{equation}
for large $\gamma$. By using the above limit in (\ref{time_in1}) and (\ref{time_fi}), for the times spent on the first and second segments, respectively, we obtain
\begin{eqnarray*}
\tau^+_i & \rightarrow & \frac{\gamma^2}{2}\cos^{-1}\left(1-\frac{2}{\gamma^2}\right)\rightarrow\frac{\gamma^2}{2}\cdot\frac{2}{\gamma}=\gamma,\\
\tau_f & \rightarrow &\frac{1}{2}\cos^{-1}(1)=0,
\end{eqnarray*}
so the total time $\tau^+_1=\tau^+_i+\tau_f$ behaves as in (\ref{previous_bound}) for large values of $\gamma$.
\end{proof}

Note that we have normalized time with the frequency $\omega_h$, thus the limiting value of the total time is actually $\gamma/\omega_h=1/\sqrt{\omega_h\omega_c}$, and we recover the value from (\ref{tau_limit}) as expected.

We now move to the case with more switchings and prove the following theorem, which is the main technical result of this work

\begin{theorem}
\label{main}
For large enough $\gamma$, there exists a positive integer $N$ in the interval
\begin{equation}
\label{N_limits}
\frac{2}{\ln5}\ln\gamma+\frac{\ln2}{\ln5}-2<N<\frac{2}{\ln5}\ln\gamma+\frac{\ln2}{\ln5}-1
\end{equation}
such that
\begin{equation}
\label{logarithmic_bound}
\tau^+_{2N+1}\rightarrow 1+\frac{1}{2}\cos^{-1}\left(\frac{3}{5}\right)+N\left\{2+\frac{1}{2}\left[\pi+\cos^{-1}\left(\frac{3}{5}\right)\right]\right\}.
\end{equation}
\end{theorem}

\begin{proof}
First of all, observe that
\begin{eqnarray*}
r_n(0)&=&\left(\frac{u_2}{u_1}\right)^{n+1}=\gamma^{4(n+1)},\\
r_n(s_m)&=&\left(\frac{s_m+u_2}{s_m+u_1}\right)^{n+1}\rightarrow 5^{n+1},
\end{eqnarray*}
where the limit corresponds to large values of $\gamma$. Now, let $N$ be the largest positive integer such that the transcendental equation (\ref{transcendentalX}) with $n=N$ and the $+$ sign has a solution. Since
\[
r_N(0)=\gamma^{4(N+1)}>\gamma^2=l_+(0)
\]
and $l_+, r_N$ are respectively increasing and decreasing functions of $s$, such an integer should satisfy
\[
r_N(s_m)<l_+(s_m)<r_{N+1}(s_m),
\]
which assures that the transcendental equation with $n=N$ has a solution while the equation with $n=N+1$ does not.
For large $\gamma$, the above inequalities take the form
\[
5^{N+1}<2\gamma^2<5^{N+2},
\]
thus such an integer always exists in this limit and is determined by the relation
\begin{equation}
\label{N}
\frac{\ln(2\gamma^2)}{\ln 5}-2<N<\frac{\ln(2\gamma^2)}{\ln 5}-1.
\end{equation}

We next find an approximation $\hat{s}$ for the unique solution $s\leq s_m$ of the transcendental equation with the $+$ sign and $n=N$, which can be compactly expressed as
\begin{equation}
\label{trans}
l_+(s)=r_N(s).
\end{equation}
Let $\hat{s}$ be the unique solution of the equation
\begin{equation}
\label{hat_equation}
r_N(\hat{s})=l_+(s_m),
\end{equation}
where note that $l_+(s_m)$ is constant. Such a $\hat{s}$ always exists for large $\gamma$, since in this limit
\[
r_N(s_m)<l_+(s_m)\rightarrow 2\gamma^2<\gamma^{4(N+1)}=r_N(0),
\]
i.e. the constant value $l_+(s_m)$ lies between the minimum and maximum values of the continuous function $r_N$. Since $l_+$ is an increasing function and $s\leq s_m$, it is also $l_+(s)\leq l_+(s_m)$ thus, combining (\ref{trans}) and (\ref{hat_equation}), we obtain
\[
r_N(s)=l_+(s)\leq l_+(s_m)=r_N(\hat{s}),
\]
which leads to the conclusion that $s\geq\hat{s}$, since $r_N$ is decreasing. Now observe that for large $\gamma$ we have
\[
\left(\frac{\hat{s}+u_{2}}{\hat{s}+u_{1}}\right)^{N+1}=r_N(\hat{s})=l_+(s_m)\rightarrow 2\gamma^2<5^{N+2},
\]
from which we obtain
\[
\hat{s}>\frac{1-5^{\left(1+\frac{1}{N+1}\right)}\cdot\frac{1}{\gamma^4}}{5^{\left(1+\frac{1}{N+1}\right)}-1}.
\]
Thus
\begin{equation}
\label{s_limits}
\frac{1-5^{\left(1+\frac{1}{N+1}\right)}\cdot\frac{1}{\gamma^4}}{5^{\left(1+\frac{1}{N+1}\right)}-1}<\hat{s}\leq s\leq s_m=\frac{1}{4}(1-\frac{1}{\gamma^4})^{2}.
\end{equation}
Obviously, in the limit of large $\gamma$, where also $N$ is large, the solution of the transcendental equation (\ref{trans}) is approaching the upper bound
\begin{equation}
\label{s_limit}
s\rightarrow\frac{1}{4}.
\end{equation}
\begin{figure}[t]
\centering
\includegraphics[width=8 cm]{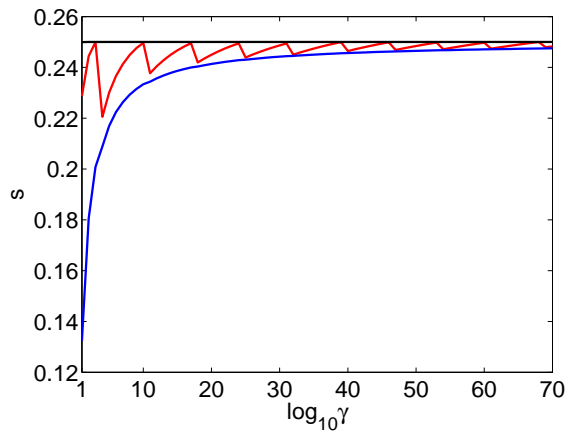}
\caption{Solutions of transcendental equation (\ref{transcendentalX}) (red line), with the $+$ sign and for the integer $n=N$ defined in (\ref{N_limits}), for various values of increasing $\gamma$. The lower (blue line) and upper (black line) bounds from (\ref{s_limits}) are also displayed. Observe that the solution converges to the limiting value $1/4$.}
\label{fig:limit}
\end{figure}
In Fig. \ref{fig:limit} we plot the solutions of transcendental equation (\ref{transcendentalX}), with the $+$ sign and for the integer $n=N$ defined in (\ref{N_limits}), for various values of increasing $\gamma$. We also display the lower and upper bounds from (\ref{s_limits}). Observe that the solution converges to the limiting value (\ref{s_limit}).

We can now use this limit to evaluate the times from Theorem \ref{solution}. We start from $\tau^+_i$. Using the relation $c_1=1+u_1$, the argument of $\cos^{-1}$ in (\ref{time_in1}) becomes
\[
\frac{sc_{1}-u_{1}\sqrt{c_{1}^{2}-4(s+u_{1})}}{(s+u_{1})\sqrt{c_{1}^{2}-4u_{1}}}=\frac{s+u_{1}\left[s-\sqrt{(1-u_1)^{2}-4s}\right]}{(s+u_{1})(1-u_1)}.
\]
For the expression within the brackets, which is already multiplied by the small quantity $u_1=1/\gamma^4$, it is $s\rightarrow 1/4$ and $\sqrt{(1-u_1)^{2}-4s}\rightarrow 0$, so we keep only the first term. The remaining part can be reformulated as
\[
\frac{s(1+u_1)}{(s+u_{1})(1-u_1)}=1+u_1\frac{2s-1+u_1}{(s+u_{1})(1-u_1)}\rightarrow 1-\frac{2}{\gamma^4},
\]
in the limit of large $\gamma$. For the same limit we obtain from (\ref{time_in1})
\[
\tau^+_i\rightarrow\frac{\gamma^2}{2}\cos^{-1}\left(1-\frac{2}{\gamma^4}\right)\rightarrow\frac{\gamma^2}{2}\cdot\frac{2}{\gamma^2}=1.
\]
Observe that the limit (\ref{s_limit}), corresponding to a slope 1/2 for the switching lines in Fig. \ref{extremal}, leads to a finite limiting value for $\tau^+_i$, contrary to the previous case with only one switching, where the limit (\ref{s_limit_0}), corresponding to a zero slope, leads to a value proportional to $\gamma$. We next find the limiting value of $\tau_{u_1}$. The argument of $\cos^{-1}$ in (\ref{switch_X}) can be written as
\[
\frac{s-u_1}{s+u_1}=1-\frac{2u_1}{s+u_1}\rightarrow 1-\frac{8}{\gamma^4},
\]
in the limit of large $\gamma$, thus
\[
\tau_{u_1}\rightarrow\frac{\gamma^2}{2}\cos^{-1}\left(1-\frac{8}{\gamma^4}\right)\rightarrow\frac{\gamma^2}{2}\cdot\frac{4}{\gamma^2}=2,
\]
i.e. it converges also to a finite value. For $\tau_{u_2}$ it is rather easy to confirm from (\ref{switch_Y}) that
\[
\tau_{u_2}\rightarrow\frac{1}{2}\left[2\pi-\cos^{-1}\left(-\frac{3}{5}\right)\right]=\frac{1}{2}\left[\pi+\cos^{-1}\left(\frac{3}{5}\right)\right].
\]
Finally, observe that for large $\gamma$ it is $\sqrt{c^{2}-4(s+u_{2})}\rightarrow c$, $\sqrt{c^{2}-4u_{2}}\rightarrow c$, and thus
\[
\frac{-sc+u_{2}\sqrt{c^{2}-4(s+u_{2})}}{(s+u_{2})\sqrt{c^{2}-4u_{2}}}\rightarrow\frac{-sc+u_2c}{(s+u_2)c}=\frac{-s+u_2}{s+u_2}\rightarrow\frac{3}{5},
\]
so from (\ref{time_fi}) we find
\[
\tau_f\rightarrow\frac{1}{2}\cos^{-1}\left(\frac{3}{5}\right).
\]

For large values of $\gamma$, the total duration of the extremal with $2N+1$ switchings approaches the value $\tau^+_{2N+1}\rightarrow\tau^+_{i}+N(\tau_{u_1}+\tau_{u_2})+\tau_{f}$ with the above calculated limiting values for the times spent on the initial, intermediate and final segments, and the limit (\ref{logarithmic_bound}) is obtained. Additionally, the bounds (\ref{N_limits}) for $N$ are derived from (\ref{N}).
\end{proof}

%%%%%%%%%%%%%%%%%%%%%%%%%%%%%%%%%%%%%%%%%%
\section{Exponential bound on the minimum achievable temperature}

Note first that Theorem \ref{main} is a clear manifestation of the unattainability principle, if the number of steps to reach absolute zero is identified with the number of switchings $2N+1$. In the limit $\omega_c\rightarrow 0$, corresponding to $T_c\rightarrow 0$, it is $\ln\gamma\rightarrow\infty$ and also $N\rightarrow\infty$ from (\ref{N_limits}), thus the number of steps becomes infinite. In this limit and by combining (\ref{logarithmic_bound}), (\ref{N_limits}), (\ref{gamma}) and (\ref{T_c}), we find that a low temperature $T_c$ can be obtained within a time $\tau$ given by
\begin{equation}
\label{min_time}
\tau\rightarrow\tau_0\ln\left(\frac{T_h}{T_c}\right),
\end{equation}
where
\begin{equation}
\label{tau0}
\tau_0=\frac{2+\frac{1}{2}\left[\pi+\cos^{-1}\left(\frac{3}{5}\right)\right]}{\ln5}\approx 2.5.
\end{equation}
Inversely, with a long process of finite duration $\tau$ can be obtained a temperature $T_c$ as low as
\begin{equation}
\label{min_temp}
T_c\rightarrow T_he^{-\frac{\tau}{\tau_0}}.
\end{equation}
This bound is obviously better than the corresponding one in (\ref{min_energy_1}) for the same system, and is also better than the bound $\tau^{-7}$ reported in \cite{Masanes17} for a general quantum system.

Here we would like to point out that a logarithmic bound for the cooling time was obtained in \cite{Hoffmann_EPL11}, by taking the limit $\gamma\rightarrow\infty$ of the minimum times already calculated in \cite{Stefanatos10,Stefanatos11}, but for the more general case where the potential is permitted to become repulsive for some time intervals, i.e. when the stiffness is allowed to vary in the range
\[
-\omega^2_h\leq\omega^2(t)\leq\omega^2_h.
\]
For the more restrictive and practically relevant case where $\omega(t)$ is bounded as in (\ref{frequency}), these bounds for time and temperature are reported here for the first time.

%%%%%%%%%%%%%%%%%%%%%%%%%%%%%%%%%%%%%%%%%%%%%%%%%%%%%
\section{OUTLOOK}

Following the above analysis arises naturally the question of whether it is possible to prove the exponential bound for more general quantum systems. We believe that it might be worthwhile an attempt of a proof along the lines of the present work. If successful, this would be another example where Mathematics can be used to further infer valuable information from the laws of Thermodynamics, the premier example being of course the proof by Carath\'{e}odory that the second law, stated in the form that not every equilibrium state of a composite system is reachable along adiabatic paths of equilibria, implies the existence of temperature and entropy functions for the composite system \cite{Caratheodory09} (stated in the modern language of differential geometry, Carath\'{e}odory proved that the above law implies that the one-form $\tilde{\delta}Q$, defined from the conservation of energy as the sum of the individual internal energies and works performed by the degrees of freedom composing the system, is integrable, i.e. there are functions $T$ (temperature) and $S$ (entropy) such that $\tilde{\delta}Q=T\tilde{d}S$ \cite{Schutz80}). We expect that the current work will contribute to the active discussion about the quantification of the third thermodynamic law.


\begin{thebibliography}{99}

\bibitem{Nernst12}
W. Nernst, Sitzber. Kgl. Preuss. Akad. Wiss. Physik-Math. Kl., 134 (1912).

\bibitem{Rezek09}
Y. Rezek, P. Salamon, K.-H. Hoffmann and R. Kosloff, EPL 85, 30008 (2009).

\bibitem{Hoffmann_EPL11}
K.-H. Hoffmann, P. Salamon, Y. Rezek, and R. Kosloff, EPL 96, 60015 (2011).

\bibitem{Levy12}
A. Levy, R. Alicki, and R. Kosloff, Phys. Rev. E 85, 061126 (2012).

\bibitem{TorronKos13}
E. Torrontegui and R. Kosloff, Phys. Rev. E 88, 032103 (2013).

\bibitem{Masanes17}
L. Masanes and J. Oppenheim, Nat. Commun. 8, 14538 (2017).

\bibitem{Stefanatos2017a}
D. Stefanatos, IEEE Trans. Automat. Control, doi: 10.1109/TAC.2017.2684083 (2017).

\bibitem{Rezek06}
Y. Rezek and R. Kosloff, New J. Phys. 8, 83 (2006).

\bibitem {Salamon09}
P. Salamon, K.-H. Hoffmann, Y. Rezek, and R. Kosloff, Phys. Chem. Chem. Phys. 11, 1027 (2009).

\bibitem{Tsirlin11}
A.M. Tsirlin, P. Salamon, and K.-H. Hoffmann, Autom. Remote Control 72, 1627 (2011).

\bibitem {Salamon12}
P. Salamon, K.-H. Hoffmann, and A. Tsirlin, Appl. Math. Lett. 25, 1263 (2012).

\bibitem{Hoffmann13}
K.-H. Hoffmann, B. Andresen, and P. Salamon, Phys. Rev. E 87, 062106 (2013).

\bibitem{Hoffmann15}
K.-H. Hoffmann, K. Schmidt, and P. Salamon, J. Non-Equilib. Thermodyn. 39, 113 (2015).

\bibitem{Boldt16}
F. Boldt, P. Salamon, and K.-H. Hoffmann, J. Phys. Chem. A 120, 3218 (2016).

\bibitem{Stefanatos2017b}
D. Stefanatos, SIAM J. Control Optim. 55, 1429 (2017).

%\bibitem{Stefanatos14PRE}
%Stefanatos, D. Optimal efficiency of a noisy quantum heat engine. {\em Phys. Rev. E} {\bf 2014}, {\em 90}, 012119.

\bibitem {Chen10}
X. Chen, A. Ruschhaupt, S. Schmidt, A. del Campo, D.
Gu\'{e}ry-Odelin, and J.G. Muga, Phys. Rev. Lett. 104, 063002 (2010).

\bibitem{Zulkowski12}
P.R. Zulkowski, D.A. Sivak, G.E. Crooks, M.R. DeWeese, Phys. Rev. E 86, 041148 (2012).

\bibitem{shortcuts13}
E. Torrontegui, S. Ib\'{a}\~{n}ez, S. Mart\'{i}nez-Garaot, M. Modugno, A. del Campo, D. Gu\'{e}ry-Odelin, A. Ruschhaupt, X. Chen, and J.G. Muga, Adv. At. Mol. Opt. Phys. 62, 117-169 (2013).

\bibitem{Deng13}
J. Deng, Q.-H. Wang, Z. Liu, P. Hanggi, and J. Gong, Phys. Rev. E 88, 062122 (2013).

\bibitem{Deffner13}
S. Deffner and E. Lutz, Phys. Rev. E 87, 022143 (2013).

\bibitem{Campo14}
A. del Campo, J. Goold, and M. Paternostro, Sci. Rep. 4, 6208 (2014).

\bibitem{Abah16}
O. Abah and E. Lutz, EPL 113, 60002 (2016).

\bibitem{Kosloff17}
Y. Rezek and R. Kosloff, Entropy 19, 136 (2017).

\bibitem{Merzbacher98}
E. Merzbacher, {\it Quantum Mechanics}, (John Wiley and Sons, New York, 1998).

\bibitem{Boldt13}
F. Boldt, J.D. Nulton, B. Andresen, P. Salamon, and K.H. Hoffmann, Phys. Rev. A 87, 022116 (2013).

\bibitem{Stefanatos10}
D Stefanatos, J Ruths, and J.-S. Li, Phys. Rev. A 82, 063422 (2010).

\bibitem{Stefanatos11}
D. Stefanatos, H. Schaettler, and J.-S. Li, SIAM J. Control Optim. 49, 2440 (2011).

\bibitem{Caratheodory09}
C. Carath\'{e}odory, Math. Ann. 67, 355 (1909).

\bibitem{Schutz80}
B. Schutz, {\it Geometrical Methods of Mathematical Physics}, (Cambridge University Press, Cambridge, 1980).





\end{thebibliography}
\end{document}